\documentclass[conference]{IEEEtran}
\IEEEoverridecommandlockouts
\usepackage{cite}
\usepackage{amsmath,amssymb,amsfonts}
\usepackage{algorithmic}
\usepackage{graphicx}
\usepackage{textcomp}
\def\BibTeX{{\rm B\kern-.05em{\sc i\kern-.025em b}\kern-.08em
    T\kern-.1667em\lower.7ex\hbox{E}\kern-.125emX}}

\usepackage{booktabs} 
\usepackage{multirow}
\usepackage{amsmath}
\usepackage{amssymb}
\usepackage{wrapfig}
\usepackage{hyperref}
\usepackage{pbox}
\usepackage{paralist}
\usepackage{algorithm}
\usepackage{algorithmic}
\usepackage{xspace}
\usepackage[font=small,skip=5pt]{caption}

\newtheorem{definition}{Definition}
\newtheorem{property}{Property}
\newtheorem{lemma}{Lemma}

\newtheorem{theorem}{Theorem}
\newtheorem{proof}{IEEEproof}

\def\tool{{CARchecker}\xspace}

\def\forward{FORWARD\xspace}
\def\backward{BACKWARD\xspace}

\def\ar{CAR\xspace}

\begin{document}

\title{Safety Model Checking with Complementary Approximations
\thanks{The full version of the ICCAD 2017 paper.}
}

\author{
    \IEEEauthorblockN{Jianwen Li\IEEEauthorrefmark{1}, Shufang Zhu\IEEEauthorrefmark{2}, Yueling Zhang\IEEEauthorrefmark{2}, Geguang Pu\thanks{Geguang Pu is the corresponding author.}\IEEEauthorrefmark{2} and Moshe Y. Vardi\IEEEauthorrefmark{1}}
    
    \IEEEauthorblockA{\IEEEauthorrefmark{1}Rice University, Houston, TX, USA
   }
   \IEEEauthorblockA{\IEEEauthorrefmark{2}East China Normal University, Shanghai, China
   }
}

\maketitle

\begin{abstract}
Formal-verification techniques, such as model checking, are becoming popular in hardware design. 
SAT-based model checking techniques, such as IC3/PDR, 
have gained a significant success in the hardware industry. 
In this paper, we present a new framework for SAT-based safety model checking, 
named \emph{Complementary Approximate Reachability} (\ar). \ar is based on standard 
reachability analysis, but instead of maintaining a single sequence of reachable-state 
sets, \ar maintains two sequences of over- and under- approximate reachable-state sets,
checking safety and unsafety at the same time. 
To construct the two sequences, \ar uses standard Boolean-reasoning algorithms,
based on satisfiability solving, one to find a satisfying cube of a satisfiable Boolean 
formula, and one to provide a minimal unsatisfiable core of an unsatisfiable Boolean 
formula.  We applied \ar to 548 hardware model-checking instances, and compared its 
performance with IC3/PDR.
Our results show that \ar is able to solve 42 instances that cannot be solved by IC3/PDR. 
When evaluated against a portfolio that includes IC3/PDR and other approaches, \ar is able to 
solve 21 instances that the other approaches cannot solve.  We conclude that \ar should be 
considered as a valuable member of any algorithmic portfolio for safety model checking.
\end{abstract}

\section{Introduction}
Model checking is a fundamental methodology in formal verification and has received 
more and more concern in the hardware design community \cite{BES16,GYS15}. Given a system model $M$ 
and a property $P$, model checking answers the question whether $P$ holds for $M$. 
When $P$ is a linear-time property, this means that we check that all behaviors 
of $M$ satisfy $P$, otherwise a violating behavior is returned as a counterexample. 
In the recent hardware model checking competition (HWMCC) \cite{BK15}, many benchmarks are collected 
from the hardware industry. Those benchmarks are modeled by the \emph{aiger} format \cite{Biere07}, in which 
the hardware circuit and properties (normally the outputs of the circuit) to be verified are both included. 
For safety checking, it answers the question whether the property (output) can be violated by feeding the 
circuit an arbitrary (finite) sequence of inputs. In this paper, we focus on the topic of safety model checking.

Popular hardware model checking techniques include Bounded Model Checking (BMC) \cite{BCCFZ99}, 
Interpolation Model Checking (IMC) \cite{McM03} and IC3/PDR \cite{Bra11,EMB11}. 
BMC reduces the search to a sequence of SAT calls, each of which corresponds to the checking in a certain step.
The satisfiability of one of such SAT calls proves the violation of the model to the given property. 
IMC combines the use of \emph{Craig Interpolation} as an abstraction technique with the use of 
BMC as a search technique.  IC3/PDR starts with an over-approximation, gradually then refined to be more 
and more precise \cite{Bra11,EMB11}. All of the three approaches have proven to be highly scalable, and are 
today parts of the algorithmic portfolio of modern symbolic model checkers, e.g.~ABC~\cite{BM10}.

We present here a new SAT-based model checking framework, named 
\emph{Complementary Approximate Reachability} (\ar), which is motivated both by classical symbolic
reachability analysis and by IC3/PDR as an abstraction-refinement technique.
While standard reachability analysis maintains a single sequence of reachable-state sets,
\ar maintains two sequences of over- and under-approximate reachable-state sets, checking 
safety and unsafety at the same time. While IC3/PDR also checks safety and unsafety at the same time,
\ar does this more directly by keeping an over-approximate sequence for safety checking,
and an under-approximate sequence for unsafety checking.  To compute these sequences, \ar utilizes 
off-the-shelf Boolean-reasoning techniques for computing \emph{Minimal Unsat Core} (MUC) \cite{MSL11},
in order to refine the over-approximate sequence, and \emph{Minimal Satisfying Cube} 
(i.e., partial assignment) \cite{YSTM14}, in order to extend the under-approximate sequences. 
In contrast, IC3/PDR uses a specialized technique, called \emph{generalization}, to compute 
\emph{Minimal Inductive Clauses} (MIC) \cite{Bra11}.  Thus, IC3/PDR computes 
\textit{relatively-inductive} clauses to refine the over- approximate state sequence, while \ar does not.
Because of this difference, \ar and IC3/PDR are complementary, with \ar faster on some problem 
instances where refining by non-relatively-inductive clauses is better, and IC3/PDR faster on others
where refining by relatively-inductive clauses is better. 

To evaluate the performance of \ar, we benchmarked it on 548 problem instances from the 2015 
Hardware Model-Checking Competition, and compared the results with IC3/ PDR. The results show 
that while the performance of \ar does not dominate the performance of IC3/PDR,
\ar complements IC3/PDR and is able to solve 42 instances that IC3/PDR cannot solve. 
When evaluated against a portfolio that includes IC3/PDR, BMC, and IMC, \ar is able to solve
21 instances that the other approaches cannot solve.
It is well known that there is no ``best'' algorithm in model checking; different 
algorithms perform differently on different problem instances \cite{ADKKM05}, and a
state-of-the-art tool must implement a portfolio of different algorithms, cf.~\cite{BM10}. 
Our empirical results also support the conclusion that \ar is an important contribution 
to the algorithmic portfolio of symbolic model checking.

The paper is organized as follows. Section \ref{sec:pre} introduces preliminaries, 
while Section \ref{sec:framework} describes the framework of \ar.  
Section \ref{sec:exp} introduces experimental results, 
and Section \ref{sec:con} discusses and concludes the paper. 

\section{Preliminaries}\label{sec:pre}
\subsection{Boolean Transition System, Safety Verification and Reachability Analysis}
A \emph{Boolean transition system} $Sys$ is a tuple $(V, I, T)$, where $V$ is a 
set of Boolean variables, and every state $s$ 
of the system is in $2^{V}$, the set of truth assignments to $V$.  
$I$ is a Boolean formula representing the set of \emph{initial} states. 
Let $V'$ be the set of primed variables (a new copy) corresponding to the variables of $V$, 
then $T$ is a Boolean formula over $V\cup V'$, denoting the transition relation 
of the system. Formally, for two states $s_1,s_2\in 2^V$, $s_2$ is a successor state of $s_1$, 
denoted as $(s_1,s_2)\in T$, iff $s_1\cup s'_2\models T$, where $s'_2$ is a primed
version of $s_2$.

A path (of length $k$) in $Sys$ is a finite state sequence $s_1,s_2,\ldots, s_k$, where each 
$(s_i,s_{i+1})(1\leq i\leq k-1)$ is in $T$. We use the notation 
$s_1\rightarrow s_2\rightarrow\ldots\rightarrow s_k$ to denote a path from $s_1$ to $s_k$. 
We say that a state $t$ is \emph{reachable from} a state $s$, or that $s$ \emph{reaches} $t$,
if there is a path from $s$ to $t$.  Moreover, we say $t$ is reachable from $s$ \emph{in $i$ steps} 
(resp., \emph{within $i$ steps}) if there is a path from $s$ to $t$ of length $i$ (resp., 
of length at most $i$). 

Let $X\subseteq 2^V$ be a set of states in $Sys$. We define 
$R(X) = \{s'|(s,s')\in T \textit{ where } s\in X\}$, i.e., 
$R(X)$ is the set of successors of states in $X$. Conversely, we define 
$R^{-1}(X) = \{s| (s, s')\in T \textit{ where } s'\in X\}$, i.e., 
$R^{-1}(X)$ is the set of predecessors of states in $X$. Recursively, we define $R^0 (X) = X$ and 
$R^i (X) = R (R^{i-1} (X))$ for $i>0$.  The notations of $R^{-i} (X)$ is defined analogously. 

Given a Boolean transition system $Sys=(V, I, T)$ and a safety property $P$, 
which is a Boolean formula over $V$, the system is called \emph{safe} if $P$ 
holds in all  reachable states of $Sys$, and otherwise it is called \emph{unsafe}. 
The safety checking asks whether $Sys$  is safe.  For unsafe systems, we want 
to find a path from an initial state to some state $s$ that violates $P$, i.e., $s\in \neg P$. 
We call such state reachable to $\neg P$ a \textit{bad} state, and the path from $I$ to $\neg P$ a \emph{counterexample}.  


In \emph{symbolic model checking} (SMC), safety checking is performed via symbolic 
reachability analysis. From the set $I$ of initial states, we compute the set of reachable
states by computing $R^i(I)$ for increasing values of $i$.  We can compute the set of states 
that can reach states in $\neg P$, by computing $(R^{-1})^i(\neg P)$ for increasing values of $i$.  
The first approach is called \emph{forward} search, while the second one is called \emph{backward} search. 
The formal definition of these two approaches are shown in the table below. 

\begin{center}
\scalebox{1}{
\begin{tabular}{|l|c|c|}
	\hline
	 & Forward & Backward\\
	 \hline
	Basic: & $F_0 = I$  &  $B	_0 = \neg P$\\
	Induction: & $F_{i+1} = R(F_i)$  & $B_{i+1} = R^{-1}(B_i)$\\
	Terminate: & $F_{i+1} \subseteq \bigcup_{0\leq j\leq i} F_j$  &  $B_{i+1} \subseteq \bigcup_{0\leq j\leq i} B_j$\\
	Check: & $F_i\cap \neg P \not=\emptyset$  &  $B_i\cap I \not= \emptyset$ \\
	\hline
\end{tabular}
}
\end{center}

For forward search, the state set $F_i$ is the set of states that are reachable 
from $I$ in $i$ steps. This set is computed by iteratively applying $R$.  To find a counterexample, 
forward search checks at every step whether one of the bad states has been reached, 
i.e., whether $F_i\cap \neg P\not=\emptyset$. If a counterexample is not found, the search will 
terminate when $F_{i+1}\subseteq \bigcup_{0\leq j\leq i} F_j$. For backward search, the set $B_i$ 
is the set of states that can reach $\neg P$ in $i$ steps. The workflow of backward search is 
analogous to that of forward search. Note that forward checking of  $Sys=(V,I,T)$ with respect to $P$
is equivalent to backward checking of $Sys^{-1}=(V,\neg P,T^{-1})$ with respect to $\neg I$,
where $T^{-1}$ is simply $T$, with primed and unprimed variables exchanged.

\subsection{Notations}

Each variable $a\in V$ is called an \textit{atom}. A \emph{literal} $l$ is an atom $a$
or a negated atom $\neg a$.  A conjunction of a set of literals, i.e., 
$l_1\wedge l_2\wedge \ldots\wedge l_k$, for $k\geq 1$, is called a \emph{cube}. 
Dually, a disjunction of a set of literals, i.e., $l_1\vee l_2\vee\ldots\vee l_k$, for $k\geq 1$, 
is called a \emph{clause}. Obviously, the negation of a cube is a clause, and vice versa. Let $C$ be a set of cubes (resp., clauses), we define the Boolean formula 
$f(C)=\bigvee_{c\in C} c$ (resp., $f(C)=\bigwedge_{c\in C}c)$. For simplicity, we use $C$ to represent 
$f(C)$ when it appears in a Boolean formula; for example, the formulas $\phi\wedge C$ and $\phi\vee C$,
abbreviate $\phi\wedge f(C)$ and $\phi\vee f(C)$.

A cube (/clause) $c$ can be treated as a set of literals, a Boolean formula, or a set of states, 
depending on the context it is used. If $c$ appears in a Boolean formula, for example, 
$c\Rightarrow \phi$, it is treated as a Boolean formula. If we say a set $c_1$ is a subset of $c_2$, 
then we treat $c_1$ and $c_2$ as literal sets. If we say a state $st$ is in $c$, then we treat $c$
as a set of states.  

We use $s(x)$/$s'(x')$ to denote the current/primed version of the state $s$. Similarly, we use $\phi(x)$/$\phi'(x')$ to denote the current/primed version of a Boolean formula $\phi$. For the transition formula $T$, we use the notation $T(x,x')$ to highlight that it contains both current and primed variables. Consider a Boolean formula $\phi$ whose alphabet is $V\cup V'$ and is in the conjunctive normal form (CNF). If $\phi$ is satisfiable, there is a full assignment $A\in 2^{V\cup V'}$ such that $A\models \phi$. Moreover, there is a \textit{partial assignment} $A^p\subseteq A$ such that for every full assignment $A'\supseteq A^p$ it holds that $A'\models \phi$. In the following, we use the notation $pa(\phi)$ to represent a partial assignment of $\phi$, and use $pa (\phi)|_x$ to represent the subset of $pa (\phi)$ achieved by projecting variables only to $V$. On the other hand, if $\phi$ is unsatisfiable, there is a \textit{Minimal Unsat Core} (MUC) $C\subseteq \phi$ (here $\phi$ is treated as a set of clauses) such that $C$ is unsatisfiable and every $C'\subset C$ is satisfiable. In the following, we use the notation $muc (\phi)$ to represent such a MUC of $\phi$, and use $muc(\phi)|_{c'}$ to represent the subset of $muc(\phi)$ achieved by projecting clauses only to $c'$. Since $c'$ is a cube, $muc(\phi)|_{c'}$ is also a cube.

\section{The Framework of \ar}\label{sec:framework}
We present here a variant of standard reachability checking, in which the set of maintained states is 
allowed to be approximate. The new approach is named \textit{Complementary Approximate Reachability}, 
abbreviated as \ar. As in standard reachability analysis, \ar also enables both forward and backward search. 
In the following, we introduce the forward approach in detail; the 
backward approach can be derived symmetrically.

\subsection{Approximate State Sequences}

In standard forward search, described in Section~\ref{sec:pre}, each $F_i$ is 
a set of states that are reachable from $I$ in $i$ steps. To compute elements in $F_{i+1}$, 
previous SAT-based symbolic-model-checking approaches consider the formula 
$\phi = F_i (x) \wedge T (x,x')$,  and use partial-assignment techniques to obtain all states in $F_{i+1}$ from $\phi$ 
(by projecting to the prime part of the assignments). Since the set of reachable states is computed accurately, maintaining a 
sequence of sets of reachable states from $I$ enables to check both safety and unsafety.
However in Forward \ar, two sequences of sets of reachable states are necessary to maintain: 1). ($F_0,F_1,\ldots$) is a sequence of over-approximate state sets, which are supersets of reachable states from $I$. 2) ($B_0,B_1,\ldots$) is a sequence of under-approximate state sets, which are subsets of reachable states to $\neg P$. Under the approximation, the first sequence is only sufficient to check safety, and the second one is then required to check unsafety.   
The two state sequences are formally defined as follows. 

\begin{definition}\label{def:approximate-states}
For a Boolean system $Sys$ and the safety property $P$, the over-approximate  
state sequences $(F_0,F_1,\ldots, F_i)$ $(i\geq 0)$, which is abbreviated as $F$-sequence, and the under-approximate state sequence $(B_0,B_1,\ldots, B_k)(k\geq 0)$, which is abbreviated as $B$-sequence, are finite sequences of state sets such that:

\begin{center}
\scalebox{0.9}{
\begin{tabular}{|l|l|r|}
\hline
Basic & $F_0 = I$ &  $B_0 = \neg P$\\
Constraint & $F_j\subseteq P(0\leq j)$ &  --\\
Inductive & $F_{j+1} \supseteq R(F_{j}) (j\geq 0)$ & $B_{j+1} \subseteq R^{-1}(B_{j})(j\geq 0)$\\
\hline
\end{tabular}
}
\end{center}

For each $F_i (i\geq 0)$, we call it a frame. We also define the notation $S(F)=\bigcup_{0\leq j\leq i} F_j$ is the set of states in the $F$-sequence, and $S(B)=\bigcup_{0\leq j\leq k}B_j$ is the set of states in the $B$-sequence.
\end{definition}

Note that the $F$- and $B$-sequence are not required to have the same length. Intuitively, each $F_{i+1}$ is an over-approximate set of states that are reachable from $F_{i}$ in one step, and $B_{i+1}$ 
is an under-approximate set of states that are reachable to $B_{i}$ in one step. 
As we mentioned in Section~\ref{sec:pre}, we overload notation and consider $F_i$ to represent 
(1) a set of states, (2) a set of clauses and (3) a Boolean formula in CNF. Analogously, we consider $B_i$ to be 
(1) a set of states, (2) a set of cubes and (3) a Boolean formula in DNF.

The following theorem shows that, the safety checking is preserved even if $F_i(i\geq 0)$ becomes 
over-approximate. 

\begin{theorem}[Safety Checking]\label{thm:unsat}
	A system $Sys$ is safe for $P$ iff there is $i\geq 0$ and an $F$-sequence $(F_0,F_1,\ldots,F_i, F_{i+1})$ such that $F_{i+1}\subseteq \bigcup_{0\leq j\leq i}F_j$. 
\end{theorem}
\begin{proof}
	($\Leftarrow$) 
	Let $S = \bigcup_{0\leq j\leq i} F_j$. According to Definition \ref{def:approximate-states}, if $F_{i+1}\subseteq S$ is true, $R(S) = R(\bigvee_{0\leq j\leq i-1}F_j)\cup R(F_{i}) \subseteq ((S\cup F_{i+1}) = S)$. So 
	$S$ contains all reachable states from $I$. Also we know $P\supseteq F_i (0\leq i\leq k)$, so $P\supseteq S$ holds. 
	That means $S\cap \neg P$ is empty, and thus all reachable states from $I$, which are included in $S$, are not in 
	$\neg P$. So the system $Sys$ is safe for $P$. 

	($\Rightarrow$) 
	Assume the system $Sys$ is safe for $P$, then all reachable states from $I$ are in $P$. Let $S\subseteq P$ be the set of reachable states from $I$. Now let $F_0 = I, F_1 = S$ 
	and $F_2 = S$, and according to Definition \ref{def:approximate-states} we know that $\delta=(F_0, F_1, F_2)$ is an 
	$F$-sequence satisfying $F_2 \subseteq F_0\cup F_1$.
\end{proof}

Theorem \ref{thm:unsat} is insufficient  for unsafety checking, as $F_{i+1}\subseteq \bigcup_{0\leq j\leq i}F_j$ has to prove 
false for every $i\geq 0$. On the other hand, the unsafety checking condition $\exists i\cdot F_i\cap \neg P\not = \emptyset$ in the standard forward reachability is not correct when $F_i$ becomes over-approximate. Our solution is to 
benefit from the information stored in the $B$-sequence.

\begin{theorem}[Unsafety Checking]\label{thm:sat}
	For a system $Sys$ and the safety property $P$, $Sys$ is unsafe for $P$ iff there is $i\geq 0$ and a $B$-sequence $(B_0,B_1,\ldots,B_i)$ such that $I\cap B_i\not = \emptyset$.
\end{theorem}
\begin{proof}
($\Leftarrow$) If $I\cap B_i\not = \emptyset$ and according to Definition \ref{def:approximate-states}, we can find a path $\rho = s_0\rightarrow s_1\rightarrow\ldots s_i$ in $Sys$ such that $s_0\in I\cap B_i$ and $s_j\in B_{i-j}$ for $1\leq j\leq i$. Hence we have that $s_i\in B_0=\neg P$, which means $\rho$ is a counterexample. So $Sys$ is unsafe for $P$.

($\Rightarrow$) If $Sys$ is unsafe for $P$, there is a path $\rho$ from $I$ to $\neg P$. Let the length of $\rho$ 
be $n+1$, and state $j (0\leq j\leq n)$ on the path is labeled as $\rho[j]$. Now we construct the $B$-sequence ($B_0,B_1,\ldots,B_i$) in the following way: Let $i=n$ and $B_j = \{\rho[i-j]\}$ for $0\leq j\leq i$. So $B_i=\{\rho[0]\}$ satisfying $B_i\cap I\not =\emptyset$ (because $\rho[0]\in I$). 
\end{proof}

Besides, since \ar maintains two different sequences, exploring the relationship between them can help to establish the framework. 
The following property shows that, the states stored in $F$- and $B$- sequences are unreachable when the system $Sys$ is safe for the property $P$.  

\begin{property}\label{coro:relation}
	For a system $Sys$ and the safety property $P$, $Sys$ is safe for $P$ iff there is an $F$-sequence such that $S(F)\cap R^{-1}(S(B))=\emptyset$ for every $B$-sequence.
\end{property}
\begin{proof}
	($\Rightarrow$) Theorem \ref{thm:unsat} shows if $Sys$ is safe for $P$ there is an $F$-sequence and $n\geq 0$ such that $F_{n+1}\subseteq \bigcup_{0\leq j\leq n}F_j$. Let $S = \bigcup_{0\leq j\leq n}F_j$. We have proven that $R(S)\subseteq S$, i.e. $S$ is the upper bound of $S(F)$. So $S(F)=\bigcup_{0\leq j\leq i}F_j= S$ for all $i\geq n$. On the other hand, since we consider arbitrary $B$-sequence, we set $S(B)$ to its upper bound: the set of all reachable states to $\neg P$. In this situation, $R^{-1}(S(B))=S(B)$. Now if $S(F)\cap R^{-1}(S(B))\not = \emptyset$, assume that $s\in S(F)\cap R^{-1}(S(B))$. So $s$ is reachable to $\neg P$. Assume $s$ is reachable to $t\in \neg P$ and from the definition of $S(F)$ we know $t\in S(F)$ too. However, this is a contradiction, because $S(F)\cap \neg P$ is empty based on the constraint $S(F)\subseteq P$. So $S(F)\cap R^{-1}(S(B))=\emptyset$ is true.
	
	($\Leftarrow$) Since both $S(F)$ and $S(B)$ have an upper bound, we can set them to their upper bounds. That is, set $S(F)$ to contain all reachable states from $I$, and $S(B)$ to contain all reachable states to $\neg P$. Because $R^{-1}(S(B))$ has the same upper bound with $S(B)$, $S(F)\cap R^{-1}(S(B))=\emptyset$ indicates that $I$ is not reachable to $\neg P$. So $Sys$ is safe for $P$.
\end{proof}

Property \ref{coro:relation} suggests a direction that how we can refine the $F$-sequence and update the $B$-sequence. That is to try to make the states in these two sequences unreachable. More details are shown in the next section.

We have established the Forward \ar framework, and presented the theoretical guarantee for both safety and unsafety checking. 
Note that symmetrically, 
Backward \ar performs the same framework on $Sys^{-1}= (V, \neg P, T^{-1})$ with respect to $\neg I$, 
where $T^{-1}$ is simply $T$ with primed and unprimed variables exchanged. 

\subsection{The Framework}

Unlike the standard forward reachability, which computes 
all states in $F_{i+1}$ from the single formula $F_i(x)\wedge T(x,x')$, 
Forward \ar computes elements of $F_{i+1}$ from different SAT calls with different inputs. 
Each SAT call gets as input a formula of the form $F_i(x)\wedge T(x,x') \wedge c'(x')$,
where the cube $c$ is in some $B_j$ and $c'$ is its primed version. If the formula is satisfiable, we are able to find a new state which is in $B_{j+1}$; otherwise we prove that $c\cap R(F_{i})=\emptyset$, which indicates $F_{i+1}$ can be refined by adding the clause $\neg c$. 
The following lemma shows the main idea of computing 
new reachable states to $\neg P$ and new clauses to refine $F_i$.
 
\begin{lemma}\label{lem:singlecompute}
Let ($F_0,F_1,\ldots$) be an $F$-sequence, ($B_0,B_1,\ldots$) be a $B$-sequence, cube $c_1\in B_j$($j\geq 0$) and the formula $\phi$ be $F_i(x)\wedge T(x,x')\wedge {c_1}'(x') (0\leq i)$:
\begin{enumerate}

\item \label{lem:singlecompute:item:1} If $\phi$ is satisfiable, there is a cube $c_2$ such that every state $t\in c_2$ is a predecessor of some state $s$ in $c_1$
and $t\in F_i$. By updating $B_{j+1} = B_{j+1}\cup \{c_2\}$, the sequence is still a $B$-sequence.

\item \label{lem:singlecompute:item:2} If $\phi$ is unsatisfiable, $c_1\cap R(F_i)=\emptyset$. Moreover, there is a cube $c_2$ such that $c_1\Rightarrow c_2$ and $c_2\cap R(F_{i})=\emptyset$. By updating $F_{i+1} = F_{i+1}\cup \{\neg c_2\}$, the sequence is still an $F$-sequence.

\end{enumerate}
\end{lemma}
\begin{proof}
	\begin{enumerate}
		\item Since $\phi$ is satisfiable, there exists a partial assignment $A^p$, which is a set of literals, of $\phi$. Now by projecting $A^p$ to the part only contains current variables, i.e. $A^p|_x$, we set $c_2=A^p|_x$. Let $t\in c_2$ and $s\in c_1$ are two states. From the definition of partial assignment, we know that $t(x)\cup s'(x')$ is a full assignment of $\phi$. So $t$ is a predecessor of $s$. Moreover, since $F_i$ consists of only current variables, $t(x)\cup s'(x')\models F_i(x)$ implies $t\models F_i$. So $t\in F_i$ is true. Before updating $B_{j+1}$, we know that $B_{j+1}\subseteq R^{-1}(B_j)$. And since $c_2\subseteq R^{-1}(c_1)$ and $c_1\subseteq B_j$, so $c_2\subseteq R^{-1}(B_j)$. Hence, $B_{j+1}\cup\{c_2\} \subseteq R^{-1}(B_j)$. That is, $B_{j+1}$ is under-approximate. Finally, for every other $B_k (1\leq (k\not =j+1))$, $B_k\subseteq R^{-1}(B_{k-1})$ is true. From Definition \ref{def:approximate-states}, the sequence is still a $B$-sequence.
		\item From the definition $R(F_i)=\{s | (t, s)\in T\textit{ and }t\in F_i\}$, we know for every state $s\in R(F_i)$ there is a state $t\in F_i$ such that $t(x)\cup s'(x')\models F_i(x)\wedge T(x,x')$. If there is a state $s\in c_1$ such that $s\in R(F_i)$, we know that there is $t\in F_i$ such that $t(x)\cup s'(x')$ is an assignment of $\phi$, making $\phi$ satisfiable. But this is a contradiction. So $c_1\cap R(F_i)=\emptyset$. Moreover, since $\phi$ is unsatisfiable, there is a cube $c_2$ such that $c_1\Rightarrow c_2$ and $F_i(x)\wedge T(x,x')\wedge {c_2}'(x')$ is still unsatisfiable. So $c_2\cap R(F_i)=\emptyset$ is also true, which implies that $\neg c_2\supseteq R(F_i)$, i.e. states represented by $\neg c_2$ includes all those in $R(F_i)$. So $F_{i+1}\cup \{\neg c_2\}\supseteq R(F_i)$ is true, which means $F_{i+1}$ is still over-approximate. For every other $F_k (1\leq (k\not = i+1))$, they remains over-approximate. As a result, we prove that the sequence is still an $F$-sequence.
	\end{enumerate}
\end{proof}

In the lemma above, Item \ref{lem:singlecompute:item:1} suggests to add a set of states rather than a single one to the $B$-sequence, and similarly Item \ref{lem:singlecompute:item:2} suggests to refine the $F$-sequence by blocking a set of states rather than a single one. In both situations, it will speed up the computation. These two kinds of heuristics can be achieved by partial-assignment and MUC techniques. That is, we can set $c_2 = pa (\phi)|_x$ in Item \ref{lem:singlecompute:item:1}, and $c_2= muc (\phi)|_{{c_1}'}$ in Item \ref{lem:singlecompute:item:2}. Now, we provide a general framework of \ar, which is shown in Table \ref{tab:framework}. 

\begin{table*}
\caption{The Framework of Forward \ar}\label{tab:framework}
\centering
\scalebox{1.1}{
\begin{tabular}{|p{15cm}|}
\hline
\begin{enumerate}
\item Initially, set $B_0=\neg P, F_0 = I$;
\item\label{framework:unsafeone} If $F_0\cap B_0 \not = \emptyset$ or $R(F_0)\cap B_0\not =\emptyset$, return unsafe with counterexample;
\item For $i\geq 1$,
\begin{enumerate}
	\item\label{framework:extend_Fi} Set $F_i := P$; 
	\item\label{framework:while} while $S(F)\cap R^{-1}(S(B))\not =\emptyset$
	\begin{enumerate}
		\item\label{framework:minimal_j} Let $j$ be the minimal index such that $F_j\cap R^{-1}(B_k)\not =\emptyset$ for some $k\geq 0$;
		\item\label{framework:unsafe} If $j=0$, return unsafe with counterexample;
		\item\label{framework:some_c} Let cube $c_1 = pa(F_j(x)\wedge T(x,x')\wedge {B_k}'(x'))|_x$ (From \ref{framework:minimal_j} $c_1$ must exist);
		\item\label{framework:add_c} Set $B_{k+1} := B_{k+1}\cup \{c_1\}$ if $B_{k+1}$ exists, otherwise set $B_{k+1}:= \{c_1\}$;
		\item Let $\phi = F_{j-1}(x)\wedge T(x,x')\wedge {c_1}'(x')$;
		\item\label{framework:sat} If $\phi$ is satisfiable, let $c_2=pa (\phi)|_x$ then assert $c_2\not\subseteq R^{-1}(S(B))$ and set $B_{k+2} := B_{k+2}\cup \{c_2\}$ if $B_{k+2}$ exists, otherwise set $B_{k+2}:= \{c_2\}$;
		\item\label{framework:unsat} If $\phi$ is unsatisfiable, let $c_2 = muc (\phi)|_{{c_1}'}$ then assert $\neg c_2 \not\supseteq F_j$ and set $F_{j} := F_{j}\cup \{\neg c_2\}$.
	\end{enumerate}
	\item\label{framework:safe} If $\exists 0\leq j\leq i\cdot F_{j}\subseteq \bigcup_{0\leq m\leq j-1} F_m$, return safe;
	\item Set $i = i+1$;
\end{enumerate}
\end{enumerate}\\
\hline
\end{tabular}
}
\end{table*}

The motivation of the computation are simply twofold: 1) Enlarge the lengths of the $F$- and $B$-sequences step by step (controlled by $i$ in the framework); 2) For each $i$, update both sequences until either the unsafety is detected (Step \ref{framework:unsafe}) or $S(F)\cap R^{-1}(B(F))=\emptyset$.  From Property \ref{coro:relation}, $S(F)\cap R^{-1}(S(B))$ is a necessary condition to prove safety (in Step \ref{framework:safe}). In Step \ref{framework:minimal_j}, we choose the minimal index because \ar aims to find a counterexample, if exists, as soon as possible. The $F$- and $B$-sequence are not extended synchronously: In each $i$, the $F$-sequence is extended only once (in Step \ref{framework:extend_Fi}), while the $B$-sequence is extended more than once (in Step \ref{framework:add_c} and \ref{framework:sat}). In Step \ref{framework:safe}, the constraint $F_{j}\subseteq \bigcup_{0\leq m\leq j-1} F_m$ can be checked by SAT solvers with the input formula ($F_j\wedge\bigwedge_{0\leq m\leq j-1}\neg F_m$). The constraint holds iff the formula is unsatisfiable. $S(F)$ and $S(B)$ are updated by default when the $F$- and $B$-sequence are updated. 
The correctness and termination of the framework are  guaranteed by the following theorem.

\begin{theorem}\label{thm:terminate}
Given a system $Sys$ and a safety property $P$, the framework terminates with a correct result.
\end{theorem}

We first show that the assertions in the framework are always true, and then prove the \textit{while} loop in Step \ref{framework:while} can finally terminate.

\begin{lemma}\label{lem:assertion}
	The assertions in Step \ref{framework:sat} and \ref{framework:unsat} are always true.
\end{lemma}
\begin{proof}
We first prove the assertion in Step \ref{framework:sat} is true. From Step \ref{framework:minimal_j} we know that $F_{j-1}\cap R^{-1}(S(B))=\emptyset$, so $(c_2\subseteq F_{j-1})\cap R^{-1}(S(B))=\emptyset$ is also true. Thus the assertion $c_2\not\subseteq R^{-1}(S(B))$ is true. For the assertion in Step \ref{framework:unsat}, first from Step \ref{framework:some_c} we know that $c_1\cap F_j\not = \emptyset$, so $(c_2\supseteq c_1)\cap F_j\not = \emptyset$ is also true. Thus the assertion $\neg c_2\not\supseteq F_j$ is true.
\end{proof}

Informally speaking, Lemma \ref{lem:assertion} guarantees that adding $c_2$ to $B_{k+1}$ increases strictly the states in $R^{-1}(S(B))$ (recall that each $B_i$ is in DNF), while adding $\neg c_2$ to $F_j$ decreases strictly the states in $F_j$ ($F_j$ is in CNF). The assertions help to prove the termination of Step \ref{framework:while}. 

\begin{lemma}\label{lem:assertion_terminate}
	Step \ref{framework:while} in the framework will finally terminate.
\end{lemma}

\begin{proof}
	Assume that $S(F)\cap R^{-1}(S(B))=\emptyset$ never holds, and $j>0$ in the framework is always true. So there is always a formula $\phi=F_{j-1}(x)\wedge T(x,x')\wedge c_1(x')$ such that the $F$- or $B$-sequence can be updated (see Step \ref{framework:sat} and \ref{framework:unsat}). Moreover, Lemma \ref{lem:assertion} guarantees that the size of each $F_i$ decreases strictly while the size of $R^{-1}(S(B))$ increases strictly. However, since the sizes of $F_i$s and $R^{-1}(S(B))$ are bounded, they cannot be updated forever. As a result, either $j=0$ must be finally true, or $S(F)\cap R^{-1}(S(B))=\emptyset$ finally holds.
\end{proof}

Now we are ready to prove Theorem \ref{thm:terminate}.

\begin{proof}
	First of all, if the framework returns, Lemma \ref{lem:singlecompute} guarantees the $F$- and $B$- sequences are preserved under the framework. Hence, the correctness is guaranteed by Theorem \ref{thm:sat} and Theorem \ref{thm:unsat}. 
	
	Now we prove the framework will finally return. First Theorem \ref{thm:terminate} guarantees the \textit{while} loop in the framework can always terminate, with unsafe or $S(F)\cap R^{-1}(S(B))=\emptyset$ is true. As a result, the loop body on each $i$ can finally terminate. Now we prove the loop on $i$ can also terminate. If $Sys$ is unsafe for $P$ with a counterexample of length of greater than one, the framework will return finally according to Item \ref{framework:unsafe}. It is because the size of $S(B)$ is bounded, so the size of $R^{-1}(S(B))$ is also bounded. Moreover, Lemma \ref{lem:assertion} guarantees the size of $R^{-1}(S(B))$ keeps growing in each $i$, thus $R^{-1}(S(B))$ will finally contain all reachable states to $\neg P$ for some $i\geq 0$, which will include an initial state in $I=F_0$.  On the other hand, if $Sys$ is safe for $P$, since $S(F)$ is also bounded, and Lemma \ref{lem:assertion} guarantees the size of each $F_j$ keeps decreasing, so $S(F)$ will finally contain only the set of reachable states from $I$ for some $i$. At that time, our framework will return according to Item \ref{framework:safe} based on Theorem \ref{thm:unsat}.
	
\end{proof}


\subsection{Related Work}\label{sec:framework:comp}

There are two main differences between \ar and IC3/PDR.
First, IC3/PDR requires the $F$-sequence to be monotone, while \ar does not.
Because \ar keeps the $F$-sequence non-monotone, it does not require the \textit{push} 
and \textit{propagate} processes, which are necessary in IC3/PDR. A drawback for \ar is that 
additional SAT calls are needed to check safety, i.e. to find $i>0$ such that 
$F_{i+1}\subseteq \bigcup_{0\leq j\leq i}F_j$ holds. In IC3/PDR, since the $F$-sequence is monotone, 
it is easy to find such $i$ that $F_i = F_{i+1}$ syntactically. 

Another main difference between \ar and IC3/PDR is the way they refine the $F$-sequence. \ar utilizes the off-the-shelf 
MUC techniques, while IC3/PDR puts more efforts to compute MIC such that the refined $F$-sequence is still monotone. 
Moreover, MIC are relatively inductive, while clauses from MUC cannot guarantee. As a result, \ar and IC3/PDR refines the 
$F$-sequence by different kinds of clauses, and thus perform differently. Although computing relatively-inductive clauses 
is proved to be efficient in IC3/PDR, we show in the experiments that, \ar complements IC3/PDR on the instances that 
computing relatively inductive clauses is not conducive for efficient checking.

It is trivial to apply the framework of \ar in both forward and backward directions by simply
reversing the direction of the model.  Indeed, our implementation of \ar runs the forward and
backward modes in parallel.
Although in theory it is also possible to run IC3/PDR in backward mode, there is a technical issue 
that must be addressed.  In most IC3/PDR implementations, the initial states $I$ is considered 
as a single cube.  This helps to save a lot of SAT calls in the process of \textit{generalization}, 
in which the computed clause $c$ must satisfy $I\wedge c$ is unsatisfiable. 
(When $I$ is a cube it is reduced to checking the containment of $\neg c\subseteq I$.) 
But usually the set of unsafe states cannot be expressed by a single cube, which makes it more complex
to run IC3/PDR in a backward mode. Indeed, the evaluation of IC3/PDR in the backward mode is still an open topic.

\ar also maintains an under-approximate state sequence (B-sequence) to check unsafety, 
while IC3/PDR checks unsafety ``on-the-fly''. Other papers also introduced multiple 
state sequences.  The approach of ``Dual Approximated Reachability'' maintains two over-approximate 
state sequences to check safety in both forward and backward directions \cite{VGS13}.
In contrast, we maintain two complementary (over- and under-) approximate state sequences to check 
safety and unsafety at the same time.  In \cite{GI15}, states reachable from initial states are maintained 
to help to handle proof-obligation generation. In contrast, the $B$-sequence keeps states that reach bad 
states. In PD-KIND\cite{JD16}, the idea of keeping both over- and under- approximate (F- and B-) state sequences 
is also introduced, and the B-sequence is used to refine the F-sequence as well. However, CAR utilizes the F-sequence 
for safety checking and B-squence for unsafey checking, while PD-KIND utilizes the F-sequence for unsafety checking and 
another ``induction frame'' has to be introduced for the safety checking. Moreover, \ar and PD-KIND use very different
underlying techniques: \ar uses MUC and partial assignment, while PD-KIND uses interpolation and generalization. 

\section{Experiments}\label{sec:exp}

\textbf{Experimental Setup }
In this section, we report the results of the empirical evaluation.
Our (C++) model checker \emph{\tool}\footnote{ \url{https://github.com/lijwen2748/CARchecker}} runs \ar in both Forward and Backward modes,
using Minisat \cite{ES03} and Muser2 \cite{MSL11} as the SAT and MUC engines. 
The tool implements the algorithm from \cite{YSTM14} to extract partial assignments.
The performance of \textit{\tool} is tested by evaluating it on the 548 safety benchmarks 
from the 2015 Hardware-Model-Checking Competition \cite{BK15}. 

We first compared the performance of \ar with that of IC3/PDR, 
as implemented in the state-of-the-art model checker ABC \cite{BM10} (the ``pdr'' command in ABC with default parameters). 
It should be noted that, there are many variants of IC3/PDR implementations currently, in which many heuristics are 
applied to the original one \cite{GR16}. We choose ABC as the reference implementation for comparison 
because it is a standard model checker integrating several kinds of SAT-based model checking techniques. 
Moreover, a portfolio of modern model checkers consists of IC3/PDR, BMC (Bounded Model Checking), 
and IMC (Interpolation Model Checking), so we also run the experiments of BMC and IMC in ABC to explore 
the contribution of \ar compared to an existing portfolio 
(we used the ``bmc2'' and ``int'' commands in ABC with default parameters).  

We run the experiments on a compute cluster that consists of 2304 
processor cores in 192 nodes (12 processor cores per node), running at 2.83 GHz with 48GB of RAM per node. 
The operating system on the cluster is RedHat 6.0.  When we run the experiments, each tool is run on a dedicated 
node, which guarantees that no CPU or memory conflict with other jobs will occur. 
Our tool \emph{\tool} can run \ar in Forward mode, Backward mode or combined mode, 
which returns the best result from either of the approaches. 
In our experiments, memory limit was set to  8 GB and time limit (CPU time) to 1 hour. 
Instances that cannot be solved within this time limit are considered unsolved, 
and the corresponding time cost is set to be 1 hour.  We compare the model-checking results 
from \tool with those from ABC on all benchmarks, and no discrepancy is found.

\setlength{\abovecaptionskip}{5pt}
\begin{figure}[t]
\centering
\includegraphics[scale=0.35]{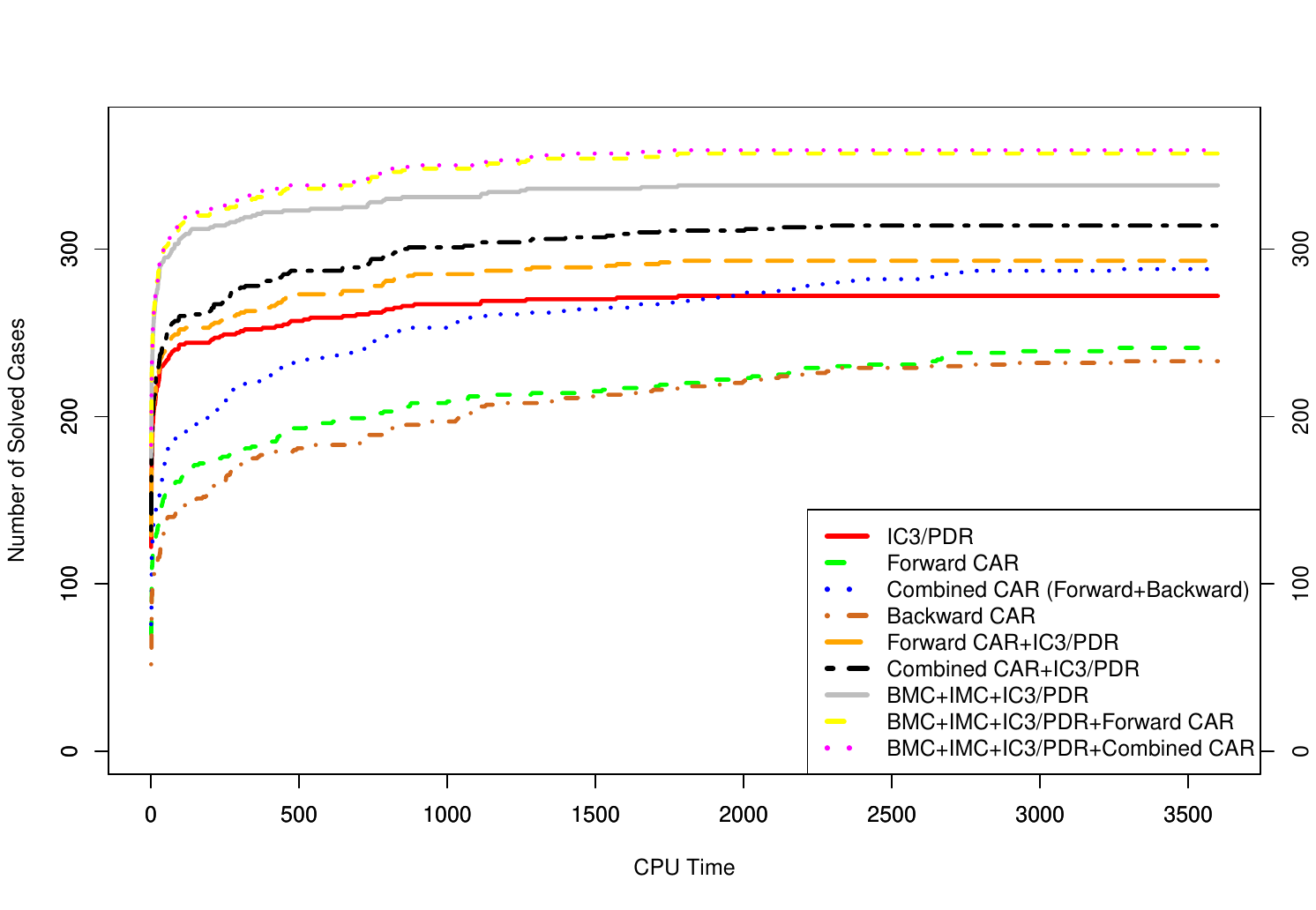}
\caption{Overall performance among different approaches.}
\label{fig:overall}
\end{figure}

\textbf{Experimental Results }
We show first overall performance comparison among different approaches in Fig. \ref{fig:overall}.
Neither Forward \ar nor Backward \ar by itself is currently competitive with IC3/PDR. The reasons are 
as follows.  First, the implementation of \tool does not utilize the power of incremental SAT computing, 
since the clauses to be added to the SAT solver are from the output of MUC solvers; 
We do not know of a way to combine them incrementally. In contrast, incremental SAT calling is 
an important feature in IC3/PDR. Secondly, ABC is a mature tool, incorporating many heuristics,
while  \emph{\tool} has only been in development for a few months so it is not be surprising that 
ABC performs better. We believe that the performance of \emph{\tool} can be improved
in the future.

\begin{figure}[t]
\centering
\includegraphics[scale=0.35]{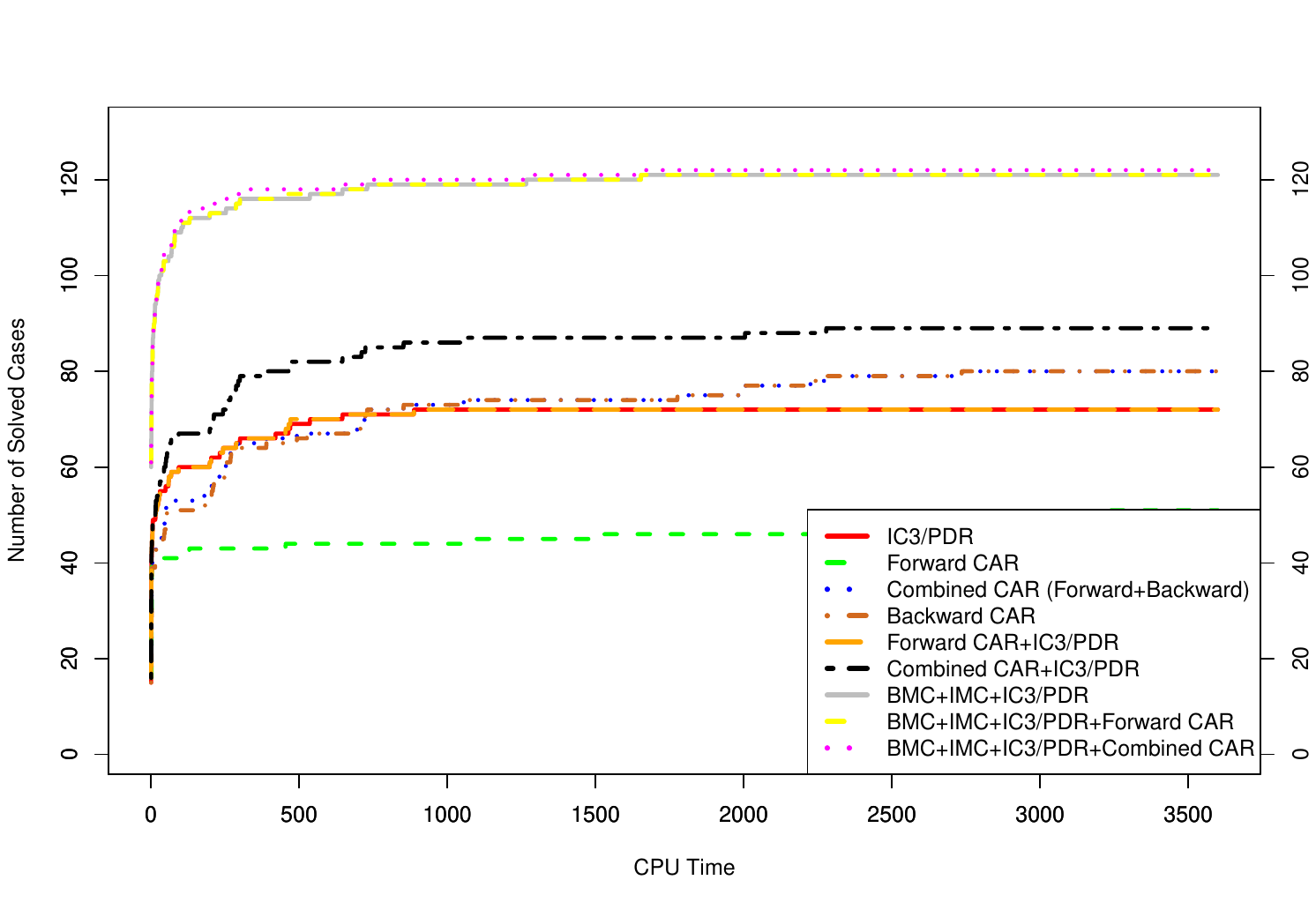}
\captionof{figure}{Overall performance on checking unsafe models.}
\label{fig:sat}
\end{figure}

Nevertheless, \ar is able to compete with IC3/PDR when combining the Forward and Backward modes. 
In Fig. \ref{fig:overall}, the plotted line for ``Combined \ar'' is obtained from the best results which selected from 
either Forward or Backward \ar: Combined \ar solves a total number of 288 instances, 
while IC3/PDR solves a total number of 271 instances. Moreover, 42 instances are solved only by Combined \ar.
We view the advantage of running \ar in both directions
as one of the contributions of this paper; it remains to be
seen whether this would also be an advantage for IC3/PDR.

\begin{figure}[t]
\centering
\includegraphics[scale=0.35]{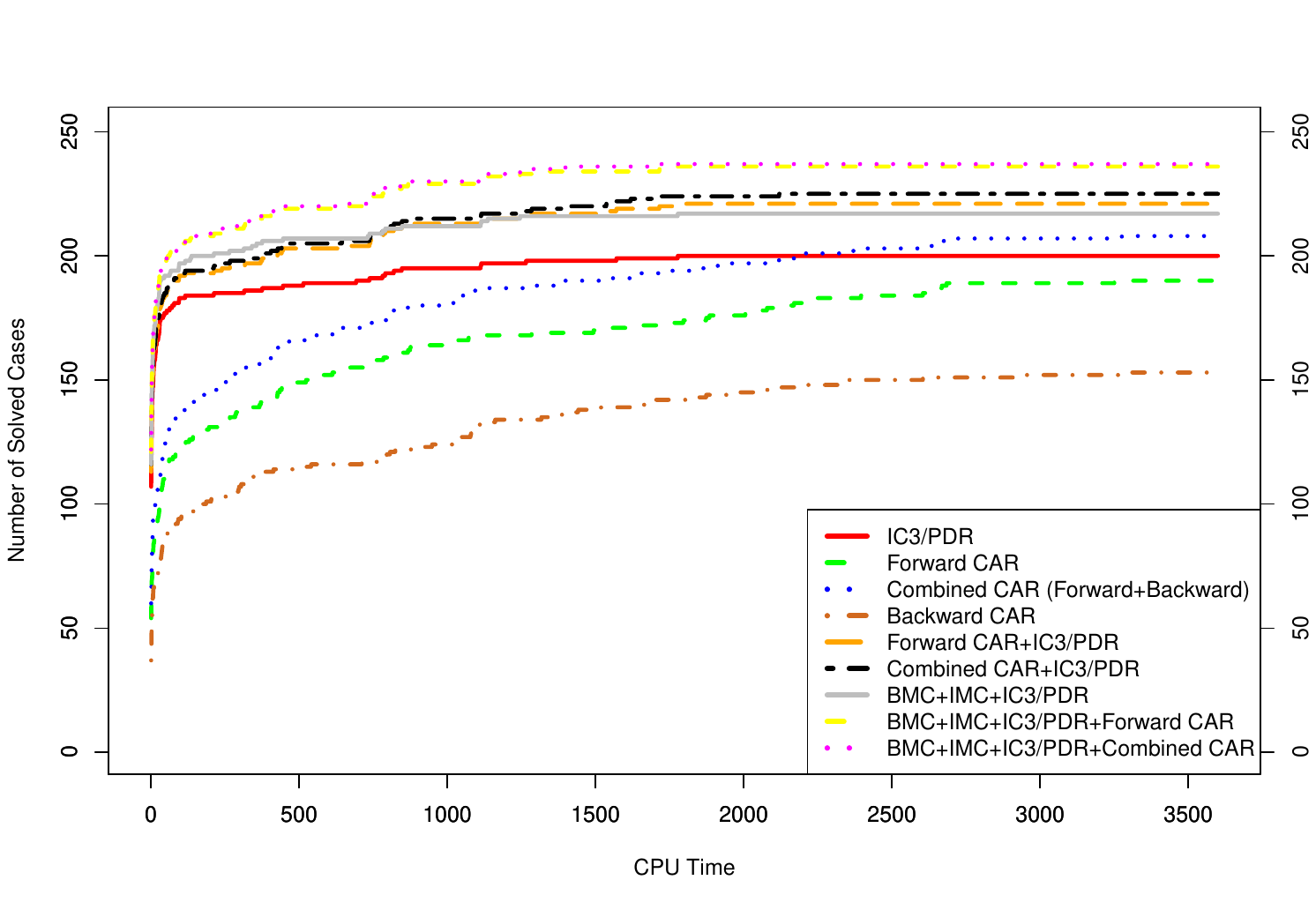}
\captionof{figure}{Overall performance on checking safe models.}
\label{fig:unsat}
\end{figure}

A major point we wish to demonstrate in this section, is that Forward \ar complements IC3/PDR. 
In Fig. \ref{fig:overall}, the plotted line for ``Forward \ar + IC3/PDR'' shows the best results from Forward \ar and 
IC3/PDR. This combination outperforms IC3/PDR by solving 21 more instances. These 21 instances are solved
by Forward \ar on checking safe models, see the results in Fig~\ref{fig:unsat}. In contrast,
in Fig. \ref{fig:sat} we see that the portfolio of Forward \ar combined with IC3/PDR does not win more instances 
than IC3/PDR on checking unsafe models (the two lines are plotted together), 
while in Fig. \ref{fig:unsat} the improvement is obvious. 

As a concrete example,
consider a real instance ``6s24.aig'' in the benchmark, which is solved quickly by Forward \ar but times out by 
IC3/PDR. Forward \ar adds the clause $\{2956\}$ (2956 is an input id in the model) into $F_{\infty}$ in Frame 1 
because the cube $\{\neg 2956\}$ is detected to represent a set of dead states. 
Moreover, it also detects the clause $\{\neg 2956\}$ must be added to 
$F_1$, as the states represented by the cube $\{2956\}$ are one step unreachable from $I$. 
As a result, Forward \ar quickly decides in Frame 1 that this model is safe, 
because both clauses $\{2956\}$ and $\{\neg 2956\}$ cannot meet in $F_1$ (recall that $F_{\infty}\subseteq F_1$). 
For IC3/PDR, although it can detect that the clause $\{2956\}$ should 
be added to $F_{\infty}$, it cannot add the clause $\{\neg 2956\}$ into $F_1$, because it is not a relatively inductive 
clause, i.e. $F_1(x)\wedge T(x,x')\wedge (\neg 2956) (x) \Rightarrow (\neg 2956)' (x')$ is not true.
As a result, computing relatively-inductive clauses is less conductive than 
computing non-relatively-inductive clauses to check this 
benchmark. So Forward \ar is able to complement IC3/PDR in such similar instances.

Furthermore, if we consider important parameters that influence the performance, 
e.g., number of clauses and number of frames, we get more positive results. 
Note that comparing the number of SAT calls between Forward \ar and PDR is not too informative, 
since Forward \ar also contains MUC calls. So fewer SAT calls in Forward \ar does not mean lower cost. 
Fig. \ref{fig:clause} and Fig. \ref{fig:frame} shows respectively the scatter plots between Forward \ar and 
IC3/PDR on number of clauses and number of frames. From the figures, Forward \ar does not generate more clauses or more 
frames than IC3/PDR. In detail, 172 (175) instances are solved with fewer clauses (frames) by Forward \ar than IC3/PDR, 
comparing with that 121 (118) instances are solved with fewer clauses (frames) by IC3/PDR than Forward \ar. 
Generally speaking, the number of clauses and frames are positively correlated to the overall performance, 
which indicates that Forward \ar should be competitive with IC3/PDR, once \tool is as optimized as ABC.

Backward \ar complements both Forward \ar and IC3/ PDR on unsafe models. 
As shown in Fig. \ref{fig:sat}, Backward \ar solves more unsafe cases (80) than IC3/PDR (72) and 
Forward \ar (51). Moreover, the combination of the three approaches solves 17 more unsafe cases 
than Forward \ar and IC3/PDR. It is surprising to see that the performance of Forward \ar is 
much worse than that of Backward \ar (51 vs 80), and all solved cases by Forward \ar are also solved by Backward \ar. 
Forward \ar searches from bad states ($\neg P$) and the goal states are in $I$, while Backward \ar searches
from $I$ and the goal states are in $\neg P$. A conjecture is that the state space of $\neg P$ is much larger than that of 
$I$, thus causing the overapproximate search to the states in $\neg P$ to be easier. 
Although Backward \ar solves more unsafe cases than IC3/PDR, there are several cases that can be solved by 
IC3/PDR but cannot be solved by Backward \ar.  We leave further comparison between IC3/PDR and Backward \ar to 
future work.

Finally, we explore the contribution of \ar to the current SAT-based model-checking portfolio, 
which includes BMC, IMC and IC3/PDR.  Fig. \ref{fig:overall}
shows the plots on the combinations IC3/PDR+BMC+IMC, IC3/PDR+BMC+IMC+Forward \ar, and IC3/PDR+BMC+IMC+Combined \ar.
Forward \ar adds 19 solved instances (all safe models) to the combination of IC3/PDR+BMC+IMC, 
and Backward \ar solves another two (1 safe and 1 unsafe model). Although BMC solves the most unsafe cases (116), 
there are three unsafe instances solved only by IC3/PDR and one unsafe instance solved only by Backward \ar. 
For safe models, the number of solved instances only by IC3/PDR, IMC, Forward \ar and Backward \ar are 13, 12, 19, 1, 
respectively.

In summary, we conclude from our experimental results that 1) Forward \ar complements IC3/PDR on checking safe models; 
2) Backward \ar complements IC3/PDR on checking unsafe models; 
3) Running \ar in both directions
improves the performance, and 4) \ar contributes to the current portfolio of model checking strategies. 
We expect these conclusions to be strengthened as the development of \emph{\tool} matures.


\begin{figure}
\centering
\includegraphics[scale=0.7]{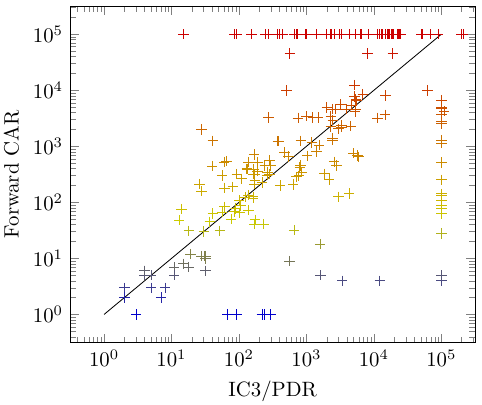}
\captionof{figure}{Comparison on clauses.}
\label{fig:clause}
\end{figure}

\begin{figure}
\centering
\includegraphics[scale=0.7]{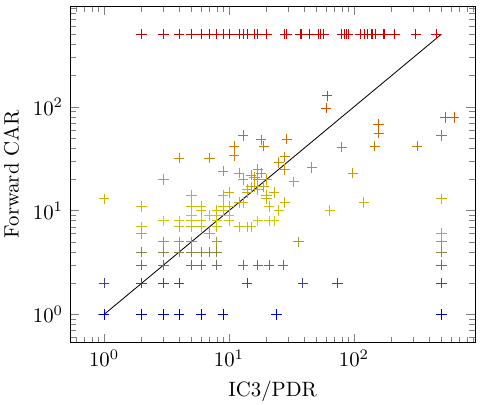}
\captionof{figure}{Comparison on frames.}
\label{fig:frame}
\end{figure}

\section{Concluding Remarks}\label{sec:con}
\ar is inspired by IC3/PDR, but it differs from it in some crucial aspects.
A main difference between \ar and IC3/PDR is that \ar does not require the $F$-sequence to be monotone. 
Also \ar uses a different strategy (MUC) to refine the $F$-sequence than IC3/PDR. 
Furthermore, \ar combines its over-approximate and under-approximate searches in both
forward and backward modes. Due to these differences, \ar and IC3/PDR have different performance
profiles.  Our experiments show that IC3/PDR and \ar complement each other. The fact that our new tool, 
after a few months of development, outperforms mature tools that have been under development
for many years over a non-negligible fraction of the benchmark suite, speaks to the merit of the
new approach.

Furthermore, the area of SAT-based model checking is still a very active research topic.
Many improvements to IC3/PDR have been proposed since the first published paper \cite{Bra11}. 
For example, a recent development is the combination of IC3/PDR with IMC in the Avy tool \cite{VG14}.
We believe that beyond the \ar tool, the \ar framework is an important contribution to this research
area and will stimulate further research. For example, it may be easier to combine IMC with \ar than
to combine IMC with IC3/PDR, as currently Avy has to pay extra effort to convert the generated interpolation 
invariants to be monotone, meeting the requirement to the state sequence maintained by IC3/PDR,
while \ar does away with this monotonicity requirement.

To conclude, we presented here \emph{Complementary Approximate Reachability} (\ar), 
a new framework for SAT-based safety model checking. \ar checks at the same time for both safety and 
unsafety in a more general way than IC3/PDR, and uses a different technique to refine the over-approximate 
state sequence.  Experiments show that \ar complements IC3/PDR and contributes to the current portfolio, 
which consists of IC3/PDR, BMC and IMC.  We argue therefore \ar is a promising approach for 
safety model checking. 

\section{Acknowledgement}
The authors thank anonymous reviewers for the helpful comments, and also thank Yakir Vizel and Alexander Ivrii for fruitful discussions. Jianwen Li is partially supported by China HGJ Project~(No. 2017ZX0 1038102-002), and NSFC Project No. 61532019. Geguang Pu is supported by NSFC Project No. 61572197. Moshe Y. Vardi is supported in part by NSF grants CCF-1319459 and IIS-1527668, and by NSF Expeditions in Computing project "ExCAPE: Expeditions in
Computer Augmented Program Engineering".

\bibliographystyle{IEEEtran}
\bibliography{ok,cav} 

\begin{thebibliography}{10}
\providecommand{\url}[1]{#1}
\csname url@samestyle\endcsname
\providecommand{\newblock}{\relax}
\providecommand{\bibinfo}[2]{#2}
\providecommand{\BIBentrySTDinterwordspacing}{\spaceskip=0pt\relax}
\providecommand{\BIBentryALTinterwordstretchfactor}{4}
\providecommand{\BIBentryALTinterwordspacing}{\spaceskip=\fontdimen2\font plus
\BIBentryALTinterwordstretchfactor\fontdimen3\font minus
  \fontdimen4\font\relax}
\providecommand{\BIBforeignlanguage}[2]{{%
\expandafter\ifx\csname l@#1\endcsname\relax
\typeout{** WARNING: IEEEtran.bst: No hyphenation pattern has been}%
\typeout{** loaded for the language `#1'. Using the pattern for}%
\typeout{** the default language instead.}%
\else
\language=\csname l@#1\endcsname
\fi
#2}}
\providecommand{\BIBdecl}{\relax}
\BIBdecl

\bibitem{BES16}
A.~Bernardini, W.~Ecker, and U.~Schlichtmann, ``Where formal verification can
  help in functional safety analysis,'' in \emph{Proceedings of the 35th
  International Conference on Computer-Aided Design}.\hskip 1em plus 0.5em
  minus 0.4em\relax New York, NY, USA: ACM, 2016, pp. 85:1--85:8.

\bibitem{GYS15}
A.~Golnari, Y.~Vizel, and S.~Malik, ``Error-tolerant processors: Formal
  specification and verification,'' in \emph{2015 IEEE/ACM International
  Conference on Computer-Aided Design (ICCAD)}, Nov 2015, pp. 286--293.

\bibitem{BK15}
A.~Biere and K.~Claessen, ``Hardware model checking competition,'' \url{.
  http://fmv.jku.at/hwmcc15/}.

\bibitem{Biere07}
A.~Biere, ``Aiger format,'' \url{http://fmv.jku.at/aiger/FORMAT}.

\bibitem{BCCFZ99}
A.~Biere, A.~Cimatti, E.~Clarke, M.~Fujita, and Y.~Zhu, ``Symbolic model
  checking using {SAT} procedures instead of {BDD}s,'' in \emph{Proc. 36st
  Design Automation Conf.}\hskip 1em plus 0.5em minus 0.4em\relax IEEE Computer
  Society, 1999, pp. 317--320.

\bibitem{McM03}
K.~McMillan, ``Interpolation and {SAT}-based model checking,'' in
  \emph{Computer Aided Verification}, ser. Lecture Notes in Computer Science,
  J.~Hunt, WarrenA. and F.~Somenzi, Eds.\hskip 1em plus 0.5em minus 0.4em\relax
  Springer Berlin Heidelberg, 2003, vol. 2725, pp. 1--13.

\bibitem{Bra11}
A.~Bradley, ``{SAT}-based model checking without unrolling,'' in
  \emph{Verification, Model Checking, and Abstract Interpretation}, ser.
  Lecture Notes in Computer Science, R.~Jhala and D.~Schmidt, Eds.\hskip 1em
  plus 0.5em minus 0.4em\relax Springer Berlin Heidelberg, 2011, vol. 6538, pp.
  70--87.

\bibitem{EMB11}
N.~Een, A.~Mishchenko, and R.~Brayton, ``Efficient implementation of property
  directed reachability,'' in \emph{Proceedings of the International Conference
  on Formal Methods in Computer-Aided Design}, 2011, pp. 125--134.

\bibitem{BM10}
R.~Brayton and A.~Mishchenko, ``{ABC}: An academic industrial-strength
  verification tool,'' in \emph{Computer Aided Verification, CAV}.\hskip 1em
  plus 0.5em minus 0.4em\relax Springer Berlin Heidelberg, 2010, pp. 24--40.

\bibitem{MSL11}
J.~Marques-Silva and I.~Lynce, ``On improving {MUS} extraction algorithms,'' in
  \emph{Theory and Applications of Satisfiability Testing - SAT 2011}, ser.
  Lecture Notes in Computer Science, K.~Sakallah and L.~Simon, Eds.\hskip 1em
  plus 0.5em minus 0.4em\relax Springer Berlin Heidelberg, 2011, vol. 6695, pp.
  159--173.

\bibitem{YSTM14}
Y.~Yu, P.~Subramanyan, N.~Tsiskaridze, and S.~Malik, ``{All-SAT} using minimal
  blocking clauses,'' in \emph{2014 27th International Conference on VLSI
  Design and 2014 13th International Conference on Embedded Systems}, 2014, pp.
  86--91.

\bibitem{ADKKM05}
N.~Amla, X.~Du, A.~Kuehlmann, R.~Kurshan, and K.~McMillan, ``An analysis of
  {SAT}-based model checking techniques in an industrial environment,'' in
  \emph{Proc. 13th IFIG Advanced Research Working Conference on Correct
  Hardware Design and Verification Methods}, ser. Lecture Notes in Computer
  Science, vol. 3725.\hskip 1em plus 0.5em minus 0.4em\relax Springer, 2005,
  pp. 254--268.

\bibitem{VGS13}
Y.~Vizel, O.~Grumberg, and S.~Shoham, ``Intertwined forward-backward
  reachability analysis using interpolants,'' in \emph{Tools and Algorithms for
  the Construction and Analysis of Systems}, 2013, pp. 308--323.

\bibitem{GI15}
A.~Gurfinkel and A.~Ivrii, ``Pushing to the top,'' in \emph{Formal Methods in
  Computer-Aided Design.}, 2015, pp. 65--72.

\bibitem{JD16}
D.~Jovanovic and B.~Dutertre, ``Property-directed k-induction,'' in
  \emph{Formal Methods in Computer-Aided Design.}, 2016, pp. 86--92.

\bibitem{ES03}
N.~E{\'e}n and N.~S{\"o}rensson, ``An extensible {SAT}-solver,'' in \emph{SAT},
  2003, pp. 502--518.

\bibitem{GR16}
A.~Griggio and M.~Roveri, ``Comparing different variants of the {IC3} algorithm
  for hardware model checking,'' \emph{IEEE Transactions on Computer-Aided
  Design of Integrated Circuits and Systems}, vol.~35, no.~6, pp. 1026--1039,
  June 2016.

\bibitem{VG14}
Y.~Vizel and A.~Gurfinkel, ``Interpolating property directed reachability,''
  \emph{Computer Aided Verification: 26th International Conference, CAV 2014},
  pp. 260--276, 2014.

\end{thebibliography}
\end{document}